\documentclass{ifacconf}

\usepackage{graphicx}      
\usepackage{amssymb,latexsym,amsfonts,amsmath,bbm}
\usepackage{dsfont}
\usepackage{mathrsfs}
\usepackage{mathtools}
\usepackage{subfig}
\usepackage{xfrac}
\usepackage{xcolor}
\usepackage{xspace}
\DeclareFontFamily{U}{stix2bb}{}
\DeclareFontShape{U}{stix2bb}{m}{n} {<-> stix2-mathbb}{}

\NewDocumentCommand{\stixbbdigit}{m}{%
	\text{\usefont{U}{stix2bb}{m}{n}#1}%
}
\newcommand{\bbzero}{\stixbbdigit{0}}
\usepackage{natbib}        



\newtheorem{thm}{Theorem}
\newtheorem{cor}{Corollary}
\newtheorem{lem}{Lemma}
\newtheorem{claim}{Claim}
\newtheorem{axiom}{Axiom}
\newtheorem{conj}{Conjecture}
\newtheorem{fact}{Fact}
\newtheorem{hypo}{Hypothesis}
\newtheorem{assum}{Assumption}
\newtheorem{prop}{Proposition}
\newtheorem{crit}{Criterion}

\newtheorem{defn}{Definition}
\newtheorem{exmp}{Example}
\newtheorem{rem}{Remark}
\newtheorem{prob}{Problem}
\newtheorem{prin}{Principle}

\newtheorem{alg}{Algorithm}

\newcommand{\R}{{\mathbb{R}}}
\newcommand{\N}{{\mathbb{N}}}

\newcommand{\Diag}[2]{%
	\underset{ #1\in\N^+}{\operatorname{diag}}(#2)%
}
\newcommand{\conblock}{\begin{bmatrix} P_i^{-1}\\ K_i\end{bmatrix}}
\newcommand{\MATLAB}{\textsc{Matlab}\xspace}

\definecolor{electricviolet}{rgb}{0.56, 0.0, 1.0}
\definecolor{forestgreen(web)}{rgb}{0.13, 0.55, 0.13}
\definecolor{royalfuchsia}{rgb}{0.79, 0.17, 0.57}

\begin{document}
\begin{frontmatter}

\title{Data-Driven Stabilizing Controller Design for Linear Infinite Networks\thanksref{footnoteinfo}}

\thanks[footnoteinfo]{A. Mironchenko has been supported by the Heisenberg grant (MI1886/3-1) of the German Research Foundation (DFG).}

\author[First]{Mahdieh Zaker}, 
\author[Second]{Andrii Mironchenko}, 
\author[First]{Amy Nejati},
\author[First]{Abolfazl Lavaei}

\address[First]{School of Computing, Newcastle University, United Kingdom (e-mails: \{m.zaker2, amy.nejati, abolfazl.lavaei\}@newcastle.ac.uk)}
\address[Second]{Department of Mathematics, University of Bayreuth, Germany (e-mail: andrii.mironchenko@uni-bayreuth.de)}

\begin{abstract}
We propose a direct data-driven method for controller synthesis of infinite networks composed of unknown linear time-invariant subsystems. {Using a single set of noise-corrupted input-state trajectories collected from each subsystem,} and provided that certain linear matrix inequalities hold, each subsystem is rendered exponentially input-to-state stable (eISS) by locally constructing an eISS control Lyapunov function together with an exponentially input-to-state stabilizing feedback controller. We then compose these local components under a compositional small-gain condition in infinite-dimensional spaces to obtain a \emph{global} control Lyapunov function and an associated stabilizing controller, ensuring uniform global exponential stability of the infinite network. The approach is validated on a physical case study with unknown dynamics.
\end{abstract}


\end{frontmatter}

\section{Introduction}

Modern engineered systems increasingly operate as large-scale interconnected networks, such as transportation systems, power grids, and cyber-physical platforms. Ensuring stability and safety in such networks remains challenging due to the scalability limitations of existing analysis and control methods. In many practical applications, including connected vehicles and smart cities, subsystems may dynamically join or leave the network, leading to uncertain and time-varying sizes~\citep{bamieh2002distributed,noroozi2022relaxed}. As a result, finite-dimensional models with a prescribed state dimension may not adequately represent the evolving structure and operating conditions of such systems. This motivates the concept of \emph{infinite} networks, in which finite networks of unknown or varying size are over-approximated by countably infinite interconnections of finite-dimensional subsystems~\citep{dashkovskiy2019stability}. The behavior of such systems is governed by subsystem interactions, through which local disturbances may propagate and cause cascading effects, making input-to-state stability (ISS) a fundamental tool for analyzing these phenomena~\citep{Son08}.

Stability results for finite ISS networks do not directly extend to infinite interconnections; for instance, infinite cascades of ISS systems may fail to preserve input-to-state stability. Leveraging advances in infinite-dimensional ISS theory~\citep{chaillet2023iss,karafyllis2019input,schwenninger2020input,mironchenko2020input} and small-gain theory for finite networks~\citep{dashkovskiy2007iss,dashkovskiy2010small}, recent studies have developed small-gain conditions tailored to infinite networks~\citep{kawan2020lyapunov,dashkovskiy2019stability,KMZ23}. However, all these results rely on accurate system models, which are often unavailable in practice, motivating direct data-driven approaches that aim to infer stability guarantees directly from data~\citep{depersis2020tac}.

While promising, most data-driven stability and control methods address \emph{monolithic} systems and are inherently limited to low-dimensional settings~\citep{depersis2020tac,chen2025data,berberich2020data,taylor2021towards}. A smaller body of work combines data-driven techniques with compositional analysis for \emph{finite} interconnected networks using subsystem ISS properties~\citep{lavaei2023ISS,zaker2025hscc,samari2025data}, but these approaches neither scale to infinite networks nor handle noisy data. Only very recently have data-driven methods been proposed for \emph{infinite} networks, focusing primarily on \emph{safety} verification~\citep{aminzadeh2024compositional,zaker2025infinite}. These noise-free methods, however, rely on state-space gridding, incur high sample complexity, and do not address controller synthesis.

\textbf{Contribution.} In contrast to the aforementioned data-driven approaches that primarily target monolithic systems or finite interconnections, we propose a direct data-driven framework for controller synthesis of \emph{infinite} networks composed of unknown linear time-invariant subsystems. Using only a single set of noise-corrupted input-state trajectories collected from each subsystem, we locally construct exponential ISS (eISS) control Lyapunov functions together with exponential input-to-state stabilizing feedback laws, thereby avoiding explicit model identification and restrictive noise-free assumptions. These eISS control Lyapunov functions for subsystems and eISS controllers are then composed upon validity of an infinite-dimensional small-gain condition to obtain a global control Lyapunov function and a stabilizing controller for the overall network. Unlike the few recent data-driven methods for infinite networks that primarily focus on safety verification and rely on computationally expensive state-space gridding, our approach enables controller synthesis with formal guarantees of uniform global exponential stability using a limited amount of data.

\textbf{Notation.}
We denote by $\R$, $\R_0^+$, and $\R^+$ the sets of real, non-negative real, and positive real numbers, respectively, while $\N^+ := \{1,2,\ldots\}$ indicates the set of positive integers.
 A sequence with a countably infinite number of vectors $x_i\in\R^{n_i}$ is written as $x=(x_i)_{i\in\N^+}$. An infinite block matrix with blocks $A_{ij}$ is denoted by $A=(A_{ij})_{i,j\in\N^+}$, while a block diagonal matrix with countably many blocks is written as $B=\Diag{i}{B_i}$.
For a vector $x_i=[x_{1i}~x_{2i}~\ldots~x_{n_ii}]^\top\in\R^{n_i}$
its $p$-norm is written as $|x_i|_p$ and is given by $|x_i|_p=\left(\sum_{j=1}^{n_i}|x_{ji}|^p\right)^{\sfrac{1}{p}}$, where $|\cdot|$ denotes the absolute value of a real number.
The space $\ell^p$ denotes the Banach space of real sequences $x=(x_i)_{i\in\N^+}$ that have a finite norm $\vert x\vert_p = (\sum_{i = 1}^{\infty} \vert x_i\vert_p ^p)^{\sfrac{1}{p}}$. Note that $\ell^2$ is a Hilbert space when endowed with the canonical scalar product. The space $\ell^\infty$ is a Banach space of bounded sequences endowed with the norm $\vert x\vert_\infty \!=\! \sup_{i\in\N^+} \vert x_i\vert<\infty$. Given sets $X_i$, $i\in\N^+$, their Cartesian product is denoted by $\prod_{i\in\N^+} X_i$.

For any symmetric matrix $P$, its smallest and largest eigenvalues are denoted by $\lambda_{\min}(P)$ and $\lambda_{\max}(P)$.  
A matrix $P\in\R^{n\times n}$ is called positive definite (respectively, positive semi-definite), denoted by $P\succ 0$ (respectively, $P\succeq 0$), if it is \emph{symmetric} and all its eigenvalues are strictly positive (respectively, non-negative). For $P\succ 0$, $\sqrt{P}$ denotes the unique positive definite matrix satisfying $\sqrt{P}^2=P$. The transpose of $P$ is denoted by $P^\top$, and $\star$ denotes a transposed component in a symmetric position of a symmetric matrix.
The identity matrix of size $n\times n$ is denoted by $\mathds{I}_n$. The zero matrix in $\R^{n\times m}$ is denoted by $\bbzero_{n\times m}$. The horizontal concatenation of vectors $x_i\in\R^{n_i}$ is written as $[x_1~x_2~\dots~x_N]\in\R^{n_i\times N}$. For a matrix $P$, $\operatorname{rank}(P)$ and $\mathscr{R}(P)$ denote its rank and range, respectively.
A function $\beta:\R_0^+\to\R_0^+$ belongs to class $\mathcal K$ if it is continuous, strictly increasing, and satisfies $\beta(0)=0$. A function $\beta:\R_0^+\times\R_0^+\to\R_0^+$ belongs to class $\mathcal{KL}$ if it is continuous, and for each fixed $s$, the mapping $r\mapsto\beta(r,s)$ is of class $\mathcal K$, and for each fixed $r>0$, the mapping $s\mapsto\beta(r,s)$ is strictly decreasing and satisfies $\beta(r,s)\to 0$ as $s\to\infty$.

\section{Preliminaries}

\subsection{Finite Subsystems}

We commence with introducing the continuous-time linear time-invariant subsystems in the following definition.
\begin{defn}\label{def:ct-LTI}
Let $i \in \N^+$ be given. A continuous-time linear time-invariant (\textsf{ct-LTI}) subsystem is described by
\begin{align}\label{eq:subsystem}
	\Omega_i\!:~ \dot x_i = A_i x_i + B_i u_i + D_i w_i,
\end{align}
where ${x_i} \in \R^{n_i}$ is the state and ${u_i} \in \R^{m_i}$ is the control input to be designed as a state-feedback controller. In addition, $A_i \in \R^{n_i \times n_i}$ and $B_i \in \R^{n_i \times m_i}$ denote the system matrix and the control matrix, respectively. The adversarial input and its associated matrix are defined as
\begin{subequations}\label{eq:interconnection}
	\begin{align}
		w_i &= (w_{ij})_{j\in\mathscr N_i} \in \R^{p_i} = \prod_{j\in \mathscr N_i} \R^{n_j},\label{eq:interconnection-a}\\
		D_i &= (D_{ij})_{j\in\mathscr N_i} \in \R^{n_i \times p_i},\label{eq:D partition}
	\end{align}
\end{subequations}
where, for all $i\in \N^+$, $\mathscr N_i \subset \N^+$ denotes the \emph{finite} set of subsystems influencing $\Omega_i$ and $p_i = \sum_{j\in\mathscr N_i} n_j$. We assume that the dimension of $w_{ij}$ is equal to that of $x_j$ for all $i \in \N^+,j \in \mathscr N_i$, and that $i \notin \mathscr N_i$ for all $i \in \N^+$. Each subsystem is denoted by $\Omega_i = (A_i,B_i,D_i)$.
\end{defn}

The matrices $A_i$ and $B_i$, characterizing the system's behavior, are assumed to be unknown, whereas $D_i$, which models the weights of interconnections, is known. 
\subsection{Infinite Networks}
Given the \textsf{ct-LTI} subsystems as described in Definition~\ref{def:ct-LTI}, we define the corresponding infinite network as follows.
\begin{defn}\label{def:inf net}
	Consider subsystems $\Omega_i = (A_i,B_i,D_i)$, $i \in \N^+$, introduced in Definition~\ref{def:ct-LTI}, together with the interconnection structure in~\eqref{eq:interconnection}, subject to the constraint
	\begin{align}\label{eq:int constraint}
		w_{i j}=x_j,\quad \forall i\in\N^+, \quad \forall j\in \mathscr N_i.
	\end{align}
	The induced \emph{infinite network}, denoted by $\Omega = \mathcal I(\Omega_i)_{i\in\N^+}$, admits the dynamics
	\begin{align}\label{eq:inf network}
		\Omega\!:~ \dot x = A x + B u,
	\end{align}
	where $x = (x_i)_{i\in\N^+}\in \mathds X$ and $u = (u_i)_{i\in\N^+}\in \mathds U$. Here, $B \coloneq \Diag{i}{B_i}$, and $A\coloneq (A_{ij})_{i,j\in\N^+}$, with
	\begin{align}\label{eq:A network}
		A_{ij}&:= \begin{cases}
			A_i, &  i\in\N^+,\;j=i,\\
			D_{ij}, & i\in\N^+,\; j\in\mathscr N_i,\\
			\bbzero_{n_i\times n_j}, &  \text{otherwise},
		\end{cases}
	\end{align}    
	and the state and input spaces for the network being defined as the $\ell^2$ spaces of sequences
	\begin{align*}
		\mathds X &= \Big\{x\,\big |\, x_i \!\in\! \R^{n_i}, \vert x\vert_2 \!=\! \Big(\sum_{i = 1}^{\infty} \vert x_i\vert_2 ^2\Big)^{\sfrac{1}{2}}\!<\infty \Big\}\subset\prod_{i\in\N^+}\R^{n_i},\\
		\mathds U &= \Big\{u \,\big |\, u_i \!\in\! \R^{m_i}, \vert u\vert_2 \!=\! \Big(\sum_{i = 1}^{\infty} \vert u_i\vert_2 ^2\Big)^{\sfrac{1}{2}}\!<\infty \Big\}\subset\prod_{i\in\N^+}\R^{m_i}.\\
	\end{align*}
	The infinite network in~\eqref{eq:inf network} is denoted by $\Omega = (A, B, \mathds X, \mathds U)$.
\end{defn}

We define the space of admissible state-feedback operators as
\begin{align*}
	&\mathcal U_F=
	\big\{K:\mathds X\to \mathds U\mid K \text{ is a bounded linear operator}\big\}.
\end{align*}
For any $K\in\mathcal U_F$, the corresponding feedback controller is given by $u=Kx$.
The closed-loop infinite network under this controller takes form
\begin{align*}
	\Omega_K:\quad \dot x = (A + BK)x,
\end{align*}
which is denoted by $\Omega_K= \mathcal I_K(\Omega_i)_{i\in\N^+}$ and its corresponding tuple is represented by $\Omega_K = (A,B,K,\mathds X,\mathds U)$. Assuming that $A:\mathds X\to \mathds X$, $B:\mathds U\to \mathds X$, and $K:\mathds X\to \mathds U$ are bounded linear operators, the closed-loop infinite network $\Omega_K = (A,B,K,\mathds X,\mathds U)$ is well-posed, as $A+BK$ is a bounded operator generating a uniformly continuous semigroup. For a given initial condition $x(0)$, the corresponding solution is thus $x(t)=e^{(A+BK)t}x(0)$, $t\ge 0$, where $e^{(A+BK)}$ is an operator exponential of $A+BK$.

With the infinite network and its constituent subsystems specified via the interconnection constraint~\eqref{eq:int constraint}, we now present sufficient conditions that guarantee uniform global exponential stability of the network.

\subsection{UGES Property}

\begin{defn}\label{def:UGES}
	The closed-loop network $\Omega_K=\mathcal I_K(\Omega_i)_{i\in\N^+}$ is called uniformly globally exponentially stable (UGES)\footnote{In this paper, UGES may refer either to global exponential stability or uniformly globally exponentially stable, depending on the context.} if there exist $M, \mu>0$ such that for all $x\in \mathds X$, we have
	$$
	\vert x(t) \vert_2 \leq Me^{-\mu t}\vert x(0)\vert_2 ,\quad t\geq 0.
	$$
\end{defn}

Since the primary objective of this work is to stabilize the infinite network $\Omega = \mathcal I(\Omega_i)_{i\in\N^+}$, we formally introduce the notion of uniform global exponential stabilizability.

\begin{defn}
	If there exists \(K\in\mathcal U_F\) such that the feedback controller \(u=Kx\) renders the closed-loop infinite network UGES, then $\Omega=\mathcal I(\Omega_i)_{i\in\mathbb{N}_+}$ is said to be uniformly globally exponentially stabilizable, and $K$ is called an (exponentially) stabilizing feedback operator.
\end{defn}

Let us denote by $\phi_K(t,x)$ the solution of $\Omega_K=(A, B, K, \mathds X,\\ \mathds U)$ at time $t\in\R_0^+$ corresponding to the initial condition $x \in \mathds X$.  
For a continuous function $V: \mathds X \to \R^+_0$, we define the \emph{Lie derivative} $\dot V(x)$ (along $\phi_K$) as the \emph{upper right-hand Dini derivative} of the function $t \mapsto V(\phi_K(t,x))$ at zero:
\begin{align}\label{eq:dini V}
    \dot V(x)    &\coloneq \limsup_{s \to 0} \tfrac{1}{s}\big(V(\phi_K(s,x)) - V(x)\big).
\end{align}
We now define the notion of control Lyapunov function.
\begin{defn}\label{def:CLF}
	A continuous function $V:\mathds X\to\R^+_0$ is called a \emph{control Lyapunov function (CLF)}, if there exist a feedback controller $u=Kx$ with $K\in\mathcal U_F$, and constants $\underline{\alpha},\overline{\alpha},\kappa\in\R^+$, such that for all $x\in\mathds X$:
	\begin{subequations}\label{eq:ct-clf-conds}
		\begin{align}
			\underline{\alpha}\vert x\vert_2^2 \le V(x)\le \overline{\alpha}\vert x\vert_2^2,&\label{eq:ct-clf-1}\\
			\dot{V}(x)\le -\kappa V(x).&\label{eq:ct-clf-2}
		\end{align}
	\end{subequations}
\end{defn}
Utilizing the notion of CLF, in the following theorem, we summarize sufficient conditions for $\Omega=\mathcal I(\Omega_i)_{i\in\N^+}$ to achieve UGES, as defined in Definition~\ref{def:UGES}. This result follows by application of the comparison principle \citep[Proposition A.35]{Mironchenko2023ISS} with a linear $\alpha$ (see also \cite{haidar2022lyapunov} for more results of this type).
\begin{thm}\label{thm:CLF}
	Given an infinite network $\Omega=(A, B, \mathds X, \mathds U)$, if there exist a CLF $V$ and a state-feedback operator $K\in\mathcal U_F$ with its corresponding state-feedback controller $u=Kx$, as defined in Definition~\ref{def:CLF},
	then the closed-loop infinite network $\Omega_K$ is UGES with $M\coloneq\sqrt{\frac{\overline\alpha}{\underline\alpha}}$ and $\mu\coloneq\frac{\kappa}{2}$.
\end{thm}

Deriving a CLF that certifies UGES for an infinite-dimensional network is, in general, computationally very expensive, even if full knowledge of the system dynamics is available. An effective workaround is to render eISS properties at the subsystem level, and to infer UGES of the overall network via a small-gain-based compositional analysis. In the following definition, we formally define the \emph{eISS control Lyapunov function} for each subsystem.
\begin{defn}\label{def:ISS-Lyap}
	Given a subsystem $\Omega_i=(A_i,B_i,D_i)$ as in~\eqref{eq:subsystem}, a continuously differentiable function $V_i:\R^{n_i}\to\R^+_0$ is an
	{exponential
    ISS (eISS)} control Lyapunov function if there are constants
	$\underline{\alpha}_i,\overline{\alpha}_i,\kappa_i\in\R^+$, and $\rho_i\in\R^+_0$, such that
		\begin{subequations}
			\begin{align}\label{eq:iss-1}
				\underline{\alpha}_i\vert x_i\vert_2^2 \le V_i(x_i)\le \overline{\alpha}_i\vert x_i\vert_2^2,\quad \forall x_i\in\R^{n_i},
			\end{align}
			and for all $x_i\in\R^{n_i}$, there exists $u_i\in\R^{m_i}$, such that for all $w_i\in\R^{p_i}$,
			\begin{align}\label{eq:iss-2}
				\mathcal L V_i(x_i)
				\le -\kappa_iV_i(x_i) + \rho_i\vert w_i\vert_2^2,
			\end{align}
		\end{subequations}
		where
		\begin{align}\label{eq:Lie derivative}
			\mathcal L V_i(x_i)=\frac{\partial V_i(x_i)}{\partial x_i}(A_i x_i + B_i u_i + D_iw_i),
		\end{align}
		is the Lie derivative of $ V_i$ along the dynamics in~\eqref{eq:subsystem}.
\end{defn}

While a CLF for the infinite network $\Omega = \mathcal I(\Omega_i)_{i\in\N^+}$ can, in principle, be obtained by composing eISS control Lyapunov functions of the individual subsystems, this approach becomes impractical when the subsystem dynamics are unknown. In particular, the absence of explicit knowledge of the matrices $A_i$ and $B_i$ in~\eqref{eq:Lie derivative} prevents the direct construction of eISS control Lyapunov functions. To  address this fundamental limitation, we propose a data-driven approach  in the following section.

\section{Main Results}

\subsection{Data-Driven Approach}
Consider an open-loop infinite network $\Omega = (A, B, \mathds X, \mathds U)$ in~\eqref{eq:inf network} subject to the interconnection constraint~\eqref{eq:int constraint}. 
We collect $N \in \mathbb{N}^+$ data samples from the unknown \textsf{ct-LTI} system in~\eqref{eq:subsystem} with sampling time $\tau \in \R^+$. To do so, for each subsystem initialized at $x_i(t_0)$, we apply arbitrary admissible inputs $u_i$ over the period $[t_0, t_0+(T-1)\tau]$ and record the corresponding data matrices $\mathbf U_i$ and $\mathbf X_i$ as
\begin{subequations}\label{eq:data}
\begin{align}
	\begin{array}{llllll}
		\mathbf U_i & \!\!=\! & [u_i(t_0) & u_i(t_0 + \tau) & \dots & u_i(t_0 +(N-1)\tau)]\in\R^{m_i\times N},\\
		\mathbf X_i & \!\!=\! & [x_i(t_0) & x_i(t_0 + \tau) & \dots & x_i(t_0 + (N-1)\tau)]\in\R^{n_i\times N}.
	\end{array}
\end{align}
Since the state data $\mathbf X_i$ of all subsystems are collected, $\mathbf W_i$ can be constructed from the corresponding neighbouring state data as $\mathbf W_i=(\mathbf W_{ij})_{j\in \mathscr N_i}\in\R^{p_i\times N}$, according to~\eqref{eq:int constraint}. Specifically, we have
\begin{align*}
	\mathbf W_{ij}= \mathbf X_j,\quad \forall i\in\N^+, \quad \forall j\in \mathscr N_i,
\end{align*}
and consequently,
\begin{align}
    \begin{array}{llllll}
         \mathbf W_i & \!\!=\! & [w_i(t_0) & w_i(t_0 + \tau) & \dots & w_i(t_0 +(N-1)\tau)]\in\R^{p_i\times N}.
    \end{array}
\end{align}
The data collected in~\eqref{eq:data} constitute what we refer to as a \emph{set of input-state trajectories}. For an infinite network with heterogeneous subsystems, this is understood in a subsystem-wise sense, meaning that one finite data set is collected from each subsystem. In practice, this requirement is naturally reduced for networks with specific structure, such as homogeneous networks or networks composed of finitely many subsystem classes, where one representative data set can be collected for each class.

Let the state derivative data be denoted by
\begin{align*}
	\begin{array}{llllll}
		\hat{\mathbf X}_i^{\mathrm d} & \!\!=\! & [\dot x_i(t_0) & \dot x_i(t_0 + \tau) & \dots & \dot x_i(t_0 + (N-1)\tau)]\in\R^{n_i\times N}.
	\end{array}
\end{align*} 
Since the state derivatives $\hat{\mathbf X}_i^{\mathrm d}$ are not directly accessible, their values are obtained through numerical differentiation. Specifically, for each $k \in \{0, \dots, N-1\}$, the derivative is approximated as 
\[
\dot x_i(t_0 + k\tau) = \frac{x_i(t_0 + (k+1)\tau) - x_i(t_0 + k\tau)}{\tau} + \varepsilon_i(t_0 + k\tau),
\]
where $\varepsilon_i(t_0 + k\tau)$ denotes the resulting approximation error and is interpreted as measurement noise. Consequently, the data available for subsequent analysis satisfy $\hat{\mathbf X}_i^{\mathrm d} =  \mathbf X_i^{\mathrm d} + \mathcal E_i$, with
$$
\mathcal E_i = [\varepsilon_i(t_0)\;\; \varepsilon_i(t_0+\tau)\;\; \dots\; \; \varepsilon_i(t_0 + (N-1)\tau)] \in \R^{n_i \times N}
$$
representing the unknown noise affecting the measurements, and $\mathbf X_i^{\mathrm d}$ defined by
\begin{align}\label{eq:derivative approx data}
	(\mathbf X_i^{\mathrm d})_{k+1} = \frac{x_i(t_0+(k+1)\tau)-x_i(t_0+k\tau)}{\tau},
\end{align}
\end{subequations}
for $k\in\{0,\ldots,N-1\}$, where $(\mathbf X_i^{\mathrm d})_{k+1}$ denotes the $(k+1)$-th column of $\mathbf X_i^{\mathrm d}$. Although $\mathcal E_i$ is unknown, we impose the following assumption.

\begin{assum}\label{asmp:noise assumption}
	The measurement noise data $\mathcal E_i \in \R^{n_i \times N}$ is not available, but it is assumed to admit a known quadratic bound of the form $\mathcal E_i \mathcal E_i^\top \preceq \Psi_i \Psi_i^\top$\!, where $\Psi_i$ of appropriate dimensions is given.
\end{assum}

\begin{rem}\label{rem:noise}
	A valid choice of $\Psi_i$ ensuring that the inequality $\mathcal E_i \mathcal E_i^\top \preceq \Psi_i \Psi_i^\top$ holds  is $\Psi_i \coloneqq \sqrt{\bar{\varepsilon}N}\,\mathds I_{n_i}$, where $\bar{\varepsilon}>0$ is a known constant such that $|\varepsilon_{i,k}|_2^2 \le \bar{\varepsilon}$ for all $k\in\{1,\ldots,N\}$, with $\varepsilon_{i,k}$ denoting each column of $\mathcal E_i$.
\end{rem}
A common choice for eISS control Lyapunov functions for \textsf{ct-LTI} subsystems is \emph{quadratic} functions as $V_i(x_i) = x_i^\top P_ix_i$, with $P_i \succ 0$. We consider the state-feedback controller for each subsystem designed as $u_i = K_iP_ix_i$, with $K_i$ and $P_i$ being constant matrices of appropriate dimensions, and denote $S_i = [A_i~B_i]$. Subsequently, one can represent the closed-loop system \eqref{eq:subsystem} under $u_i$ as
\begin{align}\label{eq:cl-rep}
	A_ix_i + B_iu_i + D_iw_i=S_i\begin{bmatrix}
		\mathds I_{n_i}\\
		K_i P_i
	\end{bmatrix} x_i+ D_i w_i.
\end{align}

We make the following assumption, which establishes a condition for the existence of a feedback gain $K_i$ and a positive-definite matrix $\Lambda_i\coloneq P_i^{-1} \succ 0$, ensuring that each subsystem is eISS.

\begin{assum}\label{asmp:main con}
	Assume that there exist scalars $\kappa_i, \vartheta_i \in \R^+$, $\gamma_i \in \R_0^+$, together with constant matrices $K_i \in \R^{m_i \times n_i}$ and $\Lambda_i \in \R^{n_i \times n_i}$ satisfying $\Lambda_i\succ 0$, such that the following linear matrix inequality (LMI) holds: 
	\begin{align}\label{eq:main con}
		\begin{bmatrix}
			\mathcal Z_i & ~~~~\begin{bmatrix}
				\Lambda_i\\
				K_i
			\end{bmatrix}^\top\!+\gamma_i\tilde{\mathbf X}_i^{\mathrm d}Q_i^\top\\
			\star & -\gamma_i Q_iQ_i^\top
		\end{bmatrix} \preceq 0,
	\end{align}
	where
	$$
	\tilde{\mathbf X}_i^{\mathrm d} = \mathbf X_i^{\mathrm d} - D_i\mathbf W_i,\quad Q_i= [\mathbf X_i^\top~\mathbf U_i^\top]^\top,
	$$
	and
	$$
	\mathcal Z_i = (\vartheta_i + \kappa_i)\Lambda_i-\gamma_i({\tilde{\mathbf X}_i^{\mathrm d}\tilde{\mathbf X}_i^{\mathrm d}}^\top\!\!-\Psi_i\Psi_i^\top).
	$$
\end{assum}
\begin{rem}
	For the proposed LMI in~\eqref{eq:main con} to be feasible, it is necessary that the pair $(A_i,B_i)$ be stabilizable. However, stabilizability alone does not, by itself, guarantee the feasibility of~\eqref{eq:main con}. Indeed, feasibility is also affected by the collected data (cf. Remark~\ref{rem:rank condition}) and by the noise bound $\Psi_i\Psi_i^\top$. Accordingly, the LMI in~\eqref{eq:main con} should be interpreted as a sufficient condition for the construction of an eISS control Lyapunov function and its associated controller.
\end{rem}
By exploiting the data set in~\eqref{eq:data} together with the closed-loop representation in~\eqref{eq:cl-rep}, and considering Assumptions~\ref{asmp:noise assumption}--\ref{asmp:main con}, we establish the following theorem, which guarantees eISS of each subsystem.
	
\begin{thm}\label{thm:main}
	Consider a \textsf{ct-LTI} $\Omega_i=(A_i,B_i, D_i)$, with the closed-loop representation~\eqref{eq:cl-rep}, where matrices $A_i$ and $B_i$ are both unknown. Suppose that Assumptions~\ref{asmp:noise assumption}--\ref{asmp:main con} hold. Then $V_i(x_i) = x_i^\top P_i x_i$
	is an eISS control Lyapunov function, with $P_i=\Lambda_i^{-1}, \underline{\alpha}_i = \lambda_{\min}(P_i)$, $\overline{\alpha}_i = \lambda_{\max}(P_i)$, $ \rho_i =  \frac{\Vert\sqrt{P_i}\Vert_2^2\Vert D_i \Vert_2^2}{\vartheta_i}$, where $\|\cdot\|_2$ denotes the induced 2-norm, and $u_i =K_i P_i x_i$ is its associated eISS controller, rendering $\Omega_i$ exponentially input-to-state stable.
\end{thm}

The proof of Theorem~\ref{thm:main} is provided in the Appendix.

\begin{rem}\label{rem:rank condition}
	To preclude a structural obstruction potentially caused by a singular bottom-right block of~\eqref{eq:main con}, it follows that the matrix $Q_i Q_i^\top$ should preferably be full rank, \emph{i.e.}, $Q_i Q_i^\top \succ 0$, which is equivalent to requiring $\operatorname{rank}(Q_i)=n_i+m_i$. Since $\operatorname{rank}(Q_i)\leq \min\{n_i+m_i,\,N\}$, this condition implies the collection of \emph{at least} $N\geq n_i+m_i$ data samples, while the input is persistently exciting~\citep{willems2005note}. In practice, this lower bound may not be enough, and a larger data set is typically required to reliably satisfy the rank condition. We note that condition~\eqref{eq:main con} may still be feasible even if $Q_i Q_i^\top$ is rank-deficient. In such a case, feasibility is ensured only if the off-diagonal block in~\eqref{eq:main con} lies within the range $\mathscr{R}(Q_i)$. This requirement can be restrictive and introduces a nontrivial additional constraint, which motivates our emphasis on ensuring that $Q_i Q_i^\top$ is full rank during the data collection using a persistently exciting input.
\end{rem}

\subsection{Infinite-Dimensional Compositional Reasoning}\label{subsec:compose}
In this section, we establish a sum-type compositional condition under which the data-driven eISS control Lyapunov functions of the individual subsystems collectively induce a CLF for the infinite network $\Omega=(A,B,\mathds X,\mathds U)$. The derivation exploits the interconnection constraint~\eqref{eq:int constraint} and formulates a small-gain requirement expressed through a linear gain operator, inspired by~\cite{kawan2020lyapunov}.

Under Assumptions~\ref{asmp:noise assumption}--\ref{asmp:main con}, Theorem~\ref{thm:main} ensures that each subsystem $\Omega_i$ admits a data-driven eISS control Lyapunov function $V_i(x_i)$  and its associated eISS controller $u_i$, satisfying~\eqref{eq:iss-1}--\eqref{eq:iss-2}, with parameters $\underline{\alpha}_i,\overline{\alpha}_i,\kappa_i\in\R^+$ and $\rho_i\in\R^+_0$, all designed purely based on data. Given these parameters, we introduce the infinite diagonal matrix $\mathcal K \coloneq \Diag{i}{\kappa_i}$ and the infinite matrix $\Delta \coloneq (\delta_{ij})_{i,j\in\N^+}$, where $\delta_{ij}\coloneq \frac{\rho_i}{\underline{\alpha}_j}$ for $j\in\mathscr N_i$, and $\delta_{ij}=0$ otherwise, according to constraint~\eqref{eq:int constraint}. We further define the gain operator
\begin{align}\label{eq:gain operator}
	\Phi \coloneq \mathcal K^{-1}\Delta = (\varphi_{ij})_{i,j\in\N^+}, \quad \varphi_{ij} \coloneq \frac{\delta_{ij}}{\kappa_i}.
\end{align}

We assume that the local Lyapunov and decay parameters are uniformly bounded in the sense that there exist constants $\underline\xi,\overline\xi,\underline\kappa\in\R^+$ and $\overline\rho\in\R_0^+$ such that, for all $i\in\N^+$,
\begin{align}\label{eq:coeffs conditions}
	0<\underline\xi \le \underline{\alpha}_i \le \overline{\alpha}_i \le \overline\xi < \infty,\quad 
	\underline\kappa \le \kappa_i,\quad \rho_i \le \overline\rho,
\end{align}
In addition, we assume that the matrix $\Delta$ satisfies
\begin{align}\label{eq:Delta}
	\|\Delta\|_{1,1} := \sup_{j\in\N^+}\sum_{i\in\N^+}\delta_{ij} < \infty,
\end{align}
which is equivalent to boundedness of the operator $\Delta\!:\ell^1\to\ell^1$. Here, the double subscript indicates that the operator norm is induced by the $\ell^1$-norm on both the domain and co-domain.

Under these assumptions, the operator $\Phi$ defined in~\eqref{eq:gain operator} acts as a linear gain operator $\Phi\!:\ell^1\to\ell^1$ according to
\begin{align}\label{eq:gain operator 1}
	(\Phi x)_i = \sum_{j\in\N^+}\varphi_{ij} x_j, \quad \forall i\in\N^+,
\end{align}
and is bounded on $\ell^1$. We impose the small-gain condition that the \emph{spectral radius} of $\Phi$, denoted by $r(\Phi)$, satisfies
\begin{align}\label{eq:con spectral}
	r(\Phi) < 1.
\end{align}

Under these assumptions, we present the following theorem, which characterizes sufficient conditions for the compositional synthesis of a CLF for $\Omega$ from the data-driven eISS control Lyapunov functions of $\Omega_i$.

\begin{thm}\label{thm:comp}
		Consider the infinite network $\Omega = \mathcal I(\Omega_i)_{i\in\N^+}$ defined in Definition~\ref{def:inf net}. Suppose that, for each $i\in\N^+$, an eISS controller $u_i$ and a corresponding eISS control Lyapunov function $V_i$ are obtained from data via Theorem~\ref{thm:main} under Assumptions~\ref{asmp:noise assumption}--\ref{asmp:main con}. Assume further that conditions~\eqref{eq:coeffs conditions}--\eqref{eq:con spectral} are satisfied. 
        Then, for any $\epsilon>0$, there exists a vector $\eta = (\eta_i)_{i\in\N^+}\in \ell^\infty$ such that $\underline{\eta} \le \eta_i \le \overline{\eta}$ for all $i\in\N^+$, where $\underline{\eta} \coloneq \inf_i(\eta_i)>0$ and $\overline{\eta} \coloneq \sup_i(\eta_i)>0$, together with a constant $\kappa_\infty>0$, satisfying
		$$
		\kappa_\infty \ge (1-r(\Phi))\,\underline\kappa - \epsilon.
		$$
Furthermore,
        \begin{align}\label{eq:V composed}
			V(x) := \sum_{i\in\N^+}\eta_i V_i(x_i), \quad \text{with }x = (x_i)_{i\in\N^+} \in \mathds X,
		\end{align}
		is a CLF for the infinite network in the sense of Theorem~\ref{thm:CLF} under the fully decentralized controller $u = (u_i)_{i\in\N^+}=(K_i P_i x_i)_{i\in\N^+}\in \mathcal U_F$,
        with parameters $\underline{\alpha} := \underline{\eta}\,\underline{\xi}$, $\overline{\alpha} := \overline{\eta}\,\overline{\xi}$, and $\kappa := \kappa_\infty$, where $\underline\xi \coloneq \inf_i(\underline{\alpha}_i)$, $\overline\xi \coloneq \sup_i(\overline{\alpha}_i)$, $\underline\kappa \coloneq \inf_i(\kappa_i)$, and $\overline\rho \coloneq \sup_i(\rho_i)$. Consequently, the infinite network $\Omega_K$ is UGES.
\end{thm}

The proof of Theorem~\ref{thm:comp} follows the same arguments as those used in \cite[Theorem 4.1]{zaker2026data} and is therefore omitted.

\begin{rem}
	The local certificates may be obtained either from the data-driven LMI in Theorem~\ref{thm:main} or, for subsystems with known dynamics, from a corresponding model-based LMI. Hence, the compositional result applies to mixed networks containing both unknown and known subsystems.
\end{rem}
\section{Simulation Results}
In this section, we demonstrate the effectiveness of our data-driven framework by applying it to an infinite network of physical subsystems. The subsystem matrices $A_i$ and $B_i$ are assumed to be unknown, and eISS control Lyapunov functions together with corresponding controllers are synthesized for all subsystems based on collected data via the results of  Theorem~\ref{thm:main}, under Assumptions~\ref{asmp:noise assumption}--\ref{asmp:main con}. These subsystem-level constructions are subsequently combined to obtain a CLF for the infinite network $\Omega = \mathcal I(\Omega_i)_{i\in\N^+}$ by verifying the compositional condition~\eqref{eq:con spectral}. The main LMI~\eqref{eq:main con} is solved in \MATLAB \textsl{R2023b} using YALMIP~\citep{Lofberg2004} with the MOSEK solver~\citep{mosek}, on a MacBook Pro equipped with an Apple~M2~Pro processor and 16~GB of memory.

We apply our approach to an infinite network of identical inverted pendulums coupled by springs, adapted from~\citep{guinaldo2013distributed}, where $\mathscr N_1 = \{2\}$ and $\mathscr N_i = \{i - 1, i + 1\}$ for $i\geq 2$, forming a bidirectional line topology. Each subsystem is described by~\eqref{eq:subsystem} with
$$
A_i = \begin{bmatrix}
		0 & 1\\
		\tfrac{g}{l}-\tfrac{\mathtt{Card}(\mathscr N_i)\,k}{ml^2} & 0
	\end{bmatrix}\!\!, ~ B_i = \begin{bmatrix}
		0\\
		\tfrac{1}{ml^2}
	\end{bmatrix}\!\!,
	$$
where $\mathtt{Card}(\mathscr N_i)$ denotes the cardinality of a set $\mathscr N_i$, $g = 9.8$, $m = 1.5$, $l = 3$, and $k = 2$. Considering $w_i=(w_{ij})_{j\in\mathscr N_i}=(x_j)_{j\in\mathscr N_i}$ and $D_i=(D_{ij})_{j\in\mathscr N_i}$ according to~\eqref{eq:int constraint} and \eqref{eq:D partition}, respectively, where $D_{ij} = \begin{bmatrix}
 	0 & 0\\
 	\tfrac{k}{ml^2} & 0
 \end{bmatrix}$, the influence of other subsystems is captured by
 $$
 D_iw_i = \begin{bmatrix}
 	0\\
 	\sum_{j\in\mathscr N_i}\frac{k}{ml^2}x_{1j}
 \end{bmatrix}\!\!.
 $$
 Here, $A_i$ and $B_i$ are assumed to be unknown and are included only for completeness, whereas $D_{ij}$ is known. As $\mathtt{Card}(\mathscr N_i) \leq 2$ for all $i$, and assuming $A_i$, $B_i$ and $D_i$ are uniformly bounded, the operators $A$ and $B$ remain bounded.
 
We collect $N = 6$ samples with a sampling period of $\tau = 0.1$. Assuming each element of $\mathcal E_i$ lies within the interval $[-0.01, 0.01]$, Remark~\ref{rem:noise} yields $\Psi_i\Psi_i^\top = 0.0012\mathds I_2$. Upon choosing $\kappa_i = 1$ and $\vartheta_i = 2$, we solve the LMI in~\eqref{eq:main con} using the collected data and obtain the eISS control Lyapunov functions as
\begin{align*}
    V_i(x_i)= 7983.0889\,x_{1i}^2 + 7794.0053\, x_{1i}x_{2i} + 2470.3523\,x_{2i}^2,
\end{align*}
and the eISS controllers as
\begin{align}\label{eq:control}
	u_i = -82.1719\, x_{1i} - 45.8755\, x_{2i},
\end{align}
with $\gamma_i = 0.068$, $\underline{\alpha}_i = 453.4412$, $\overline{\alpha}_i = 10^4$, and $\rho_i = 219.4787$. {Employing the gains obtained from data, we examine whether the compositional condition~\eqref{eq:con spectral} holds.} To this end, we construct the infinite gain matrix $\Phi$ according to~\eqref{eq:gain operator}, which takes the following form:
\begin{align}\label{gain}
	\Phi = 
	\begin{bmatrix}
		0 & 0.4840 & 0 & 0 & 0 & \ldots \\
		0.4840 & 0 & 0.4840 & 0 & 0 & \ldots \\
		0 & 0.4840 & 0 & 0.4840 & 0 & \ldots \\
		0 & 0 & 0.4840 & 0 & 0.4840 & \ldots \\
		\vdots & \ddots & \ddots & \ddots & \ddots & \ddots
	\end{bmatrix}\!\!.
\end{align}
To verify the small-gain condition, we simply use the following sufficient condition:
\begin{align}\label{eq:SGC}
	r(\Phi)\leq \|\Phi\|_{1,1} = \sup_{j\in\N^+}\sum_{i\in\N^+}\varphi_{ij} < 1.
\end{align}
Exploiting the specific structure of $\Phi$ in~\eqref{gain} and the fact that all subsystems are identical, we assess condition~\eqref{eq:con spectral}
using~\eqref{eq:SGC}. We obtain the bound $r(\Phi)\leq\|\Phi\|_{1,1} \leq 0.9681 < 1$, which confirms that condition~\eqref{eq:con spectral} holds. Moreover, since $\kappa_i=1$ for all $i\in\N^+$, one has $\Delta=\Phi$. Hence, $\|\Delta\|_{1,1}=\|\Phi\|_{1,1}\le 0.9681<\infty$, which verifies condition~\eqref{eq:Delta}. As a consequence, this enables the construction of the CLF $V(x)=\sum_{i\in\N^+}\eta_i V_i(x_i)$ for $\Omega$, together with its decentralized controller $u=(u_i)_{i\in\N^+}$, thereby ensuring that the infinite network is {UGES}.
\begin{figure}[t!]
	\centering
	\subfloat[\centering Open-loop evolutions\label{fig:IP ol}]{
		\includegraphics[width=0.45\linewidth]{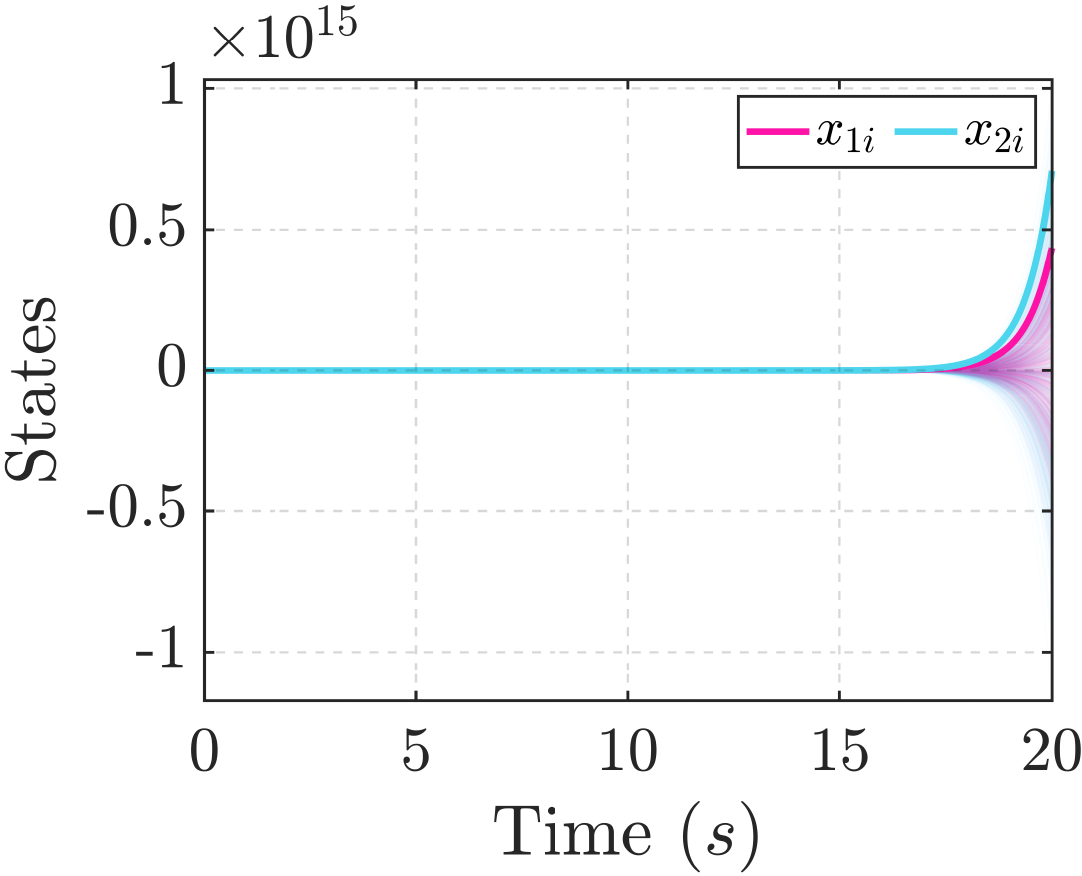}}\hfil
	\subfloat[\centering Closed-loop evolutions\label{fig:IP cl}]{
		\includegraphics[width=0.45\linewidth]{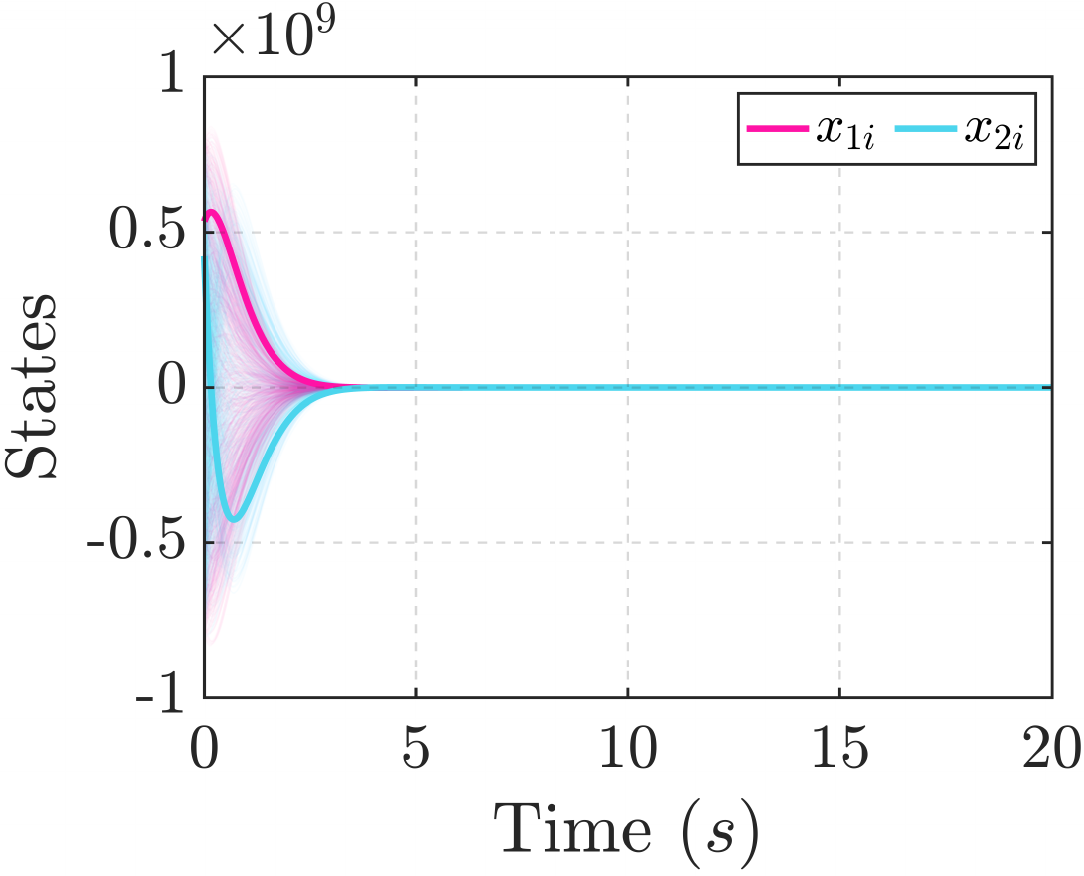}}
	\caption{Illustration of state evolution of some arbitrary subsystems with and without control, where a representative subsystem is highlighted in bold and the remaining subsystems are shown in a faded style.\label{fig:IP}
	}
\end{figure}

We apply the designed stabilizing controller to the network starting from arbitrary and large initial conditions in the interval $[-8\times10^8,8\times10^8]^2$ for finite number of subsystems and $[0~0]^\top$ for the rest. Figure~\ref{fig:IP}(b) exhibits the trajectories of some arbitrary subsystems under the designed controller~\eqref{eq:control}, while Figure~\ref{fig:IP}(a) shows that the open-loop subsystems are unstable, verifying the effectiveness of our data-driven approach.

\bibliography{ifacconf}

\section*{Appendix}

\noindent\textbf{Proof of Theorem~\ref{thm:main}.} We first verify that condition~\eqref{eq:iss-1} holds for the candidate eISS control Lyapunov function. Since
$$
	\lambda_{\min}(P_i) |x_i|_2^2 \leq \underbrace{x_i^\top P_i x_i}_{V_i(x_i)} \leq \lambda_{\max}(P_i) |x_i|_2^2,
$$
selecting $\underline{\alpha}_i=\lambda_{\min}(P_i)$ and $\overline{\alpha}_i=\lambda_{\max}(P_i)$ directly ensures~\eqref{eq:iss-1}.
We next show that condition~\eqref{eq:iss-2} holds, as well. Using the closed-loop representation~\eqref{eq:cl-rep}, we obtain
	\begin{align*}
		\mathcal L V_i(x_i) &= 2x_i^\top P_i\Big(S_i\begin{bmatrix}
			\mathds I_{n_i}\\
			K_i P_i
		\end{bmatrix} x_i+ D_i w_i\Big)\\
		&=2x_i^\top P_iS_i\begin{bmatrix}
			\mathds I_{n_i}\\
			K_i P_i
		\end{bmatrix} x_i+ 2x_i^\top P_iD_iw_i \\
		&=2x_i^\top P_iS_i\conblock P_i x_i+ 2x_i^\top P_iD_iw_i\\
		&=x_i^\top P_iS_i\conblock P_i x_i+ x_i^\top P_i\conblock^{\!\!\top} \!\!\! S_i^\top P_i x_i\\
		&\hphantom{=}~+ 2x_i^\top P_iD_iw_i\\
		& = x_i^\top \!P_i\Big(S_i\!\conblock \!+\! \conblock^{\!\!\top} \!\!\! S_i^\top\Big) P_i x_i+ 2x_i^\top P_iD_iw_i\\
		&= x_i^\top P_i\begin{bmatrix}
			\mathds I_{n_i}\\
			S_i^\top
		\end{bmatrix}^{\!\!\top}\!\!\begin{bmatrix}
			\bbzero_{n_i\times n_i} & \conblock^{\!\!\top}\\
			\star & \bbzero_{s_i\times s_i}
		\end{bmatrix}\!\!\begin{bmatrix}
			\mathds I_{n_i}\\
			S_i^\top
		\end{bmatrix} P_i x_i\\
		&\hphantom{=}~+ 2\underbrace{x_i^\top P_iD_iw_i}_{\clubsuit},
	\end{align*}
where $s_i = n_i + m_i$. We bound the term $\clubsuit$ by applying the Cauchy-Schwarz inequality~\citep{bhatia1995cauchy} in the form $ab \leq |a|_2 |b|_2$ for any $a^\top,b\in\R^n$, with the choices $a=x_i^\top\sqrt{P_i}$ and $b=\sqrt{P_i}D_iw_i$, followed by Young's inequality~\citep{young1912classes} as $|a|_2 |b|_2 \leq \frac{\vartheta_i}{2} |a|_2^2+\frac{1}{2\vartheta_i} |b|_2^2,$ for any $\vartheta_i>0$, resulting in
\begin{align*}
	\mathcal L V_i(x_i)&\leq \overbrace{x_i^\top\! P_i\!\begin{bmatrix}
		\mathds I_{n_i}\\
		S_i^\top
	\end{bmatrix}^{\!\!\top}\!\!\!\begin{bmatrix}
			\bbzero_{n_i\times n_i} & \conblock^{\!\!\top}\\
			\star & \bbzero_{s_i\times s_i}
	\end{bmatrix}\!\!\!\begin{bmatrix}
		\mathds I_{n_i}\\
		S_i^\top
	\end{bmatrix}\!\! P_i x_i}^{h_i(x_i)}\\
	&\hphantom{=}~+ \vartheta_i\underbrace{x_i^\top P_i x_i\vphantom{\frac{\Vert \sqrt{P_i}\Vert_2^2 \Vert D_i\Vert_2^2}{\vartheta_i}}}_{V_i(x_i)} + \underbrace{\frac{\Vert \sqrt{P_i}\Vert_2^2 \Vert D_i\Vert_2^2}{\vartheta_i}}_{\rho_i } |w_i|_2^2\\
	& = \underbrace{h_i(x_i) + \vartheta_iV_i(x_i)}_{\mathcal H_i(x_i)} + \rho_i|w_i|_2^2.
\end{align*}
To conclude~\eqref{eq:iss-2}, it remains to ensure that
\begin{align}\label{eq:new H}
	\mathcal H_i(x_i) \leq -\kappa_i V_i(x_i),
\end{align}
or, equivalently,
\begin{align}\label{eq:new S}
	\mathcal H_i(x_i) + \kappa_iV_i(x_i) \leq 0
\end{align}
holds for some $\kappa_i\in\R^+$. Since $S_i=[A_i~B_i]$ is unknown, this condition cannot be verified directly. However, letting $Q_i=[\mathbf X_i^\top~\mathbf U_i^\top]^\top$, the collected data in~\eqref{eq:data} satisfy
$$
	\mathbf X_i^{\mathrm d} = \underbrace{A_i\mathbf X_i + B_i\mathbf U_i + D_i\mathbf W_i}_{\hat{\mathbf X}_i^{\mathrm d}} - \mathcal E_i= S_iQ_i  + D_i\mathbf W_i - \mathcal E_i,
$$
which implies $\mathcal E_i = -(\underbrace{\mathbf X_i^{\mathrm d} - D_i\mathbf W_i}_{\tilde{\mathbf X}^{\mathrm{d}}_i}-S_iQ_i)$.
Under Assumption~\ref{asmp:noise assumption}, we have
$$
	(\tilde{\mathbf X}_i^{\mathrm d}-S_iQ_i)(\tilde{\mathbf X}_i^{\mathrm d}-S_iQ_i)^\top- \Psi_i\Psi_i^\top\preceq0,
$$
which yields
$$
	\underbrace{\begin{bmatrix}
			\mathds I_{n_i}\\
			S_i^\top
		\end{bmatrix}^{\!\!\top}\begin{bmatrix}
				{\tilde{\mathbf X}_i^{\mathrm d}\tilde{\mathbf X}_i^{\mathrm d}}^\top-\Psi_i\Psi_i^\top & -\tilde{\mathbf X}_i^{\mathrm d}Q_i^\top\\
				\star & Q_iQ_i^\top
		\end{bmatrix}\begin{bmatrix}
			\mathds I_{n_i}\\
			S_i^\top
	\end{bmatrix}}_{\Xi_i}\preceq 0.
$$
Pre-multiplying and post-multiplying $\Xi_i\in\R^{n_i\times n_i}$ by $x_i^\top P_i$ and $P_i x_i$, respectively, results in
\begin{align}\label{eq:proof-new1}
	z_i(x_i)\coloneq x_i^\top P_i \Xi_i P_i x_i\leq 0.
\end{align}
We now use the S-procedure~\citep{yakubovich2004stability} to enforce~\eqref{eq:new S} while respecting the data-consistency constraint~\eqref{eq:proof-new1}. In particular, if there exists $\gamma_i\ge 0$ such that
\begin{align}\label{eq:S-procedure}
	\mathcal H_i(x_i) + \kappa_iV_i(x_i) - \gamma_i z_i(x_i) \leq 0,
\end{align}
then~\eqref{eq:new S} holds for all $x_i$, and all $S_i$ consistent with the data requirement in~\eqref{eq:proof-new1}.
By reformulating $V_i(x_i)$ as
\begin{align*}
	V_i(x_i)& = x_i^\top P_i x_i = x_i^\top \underbrace{P_i P_i^{-1}}_{\mathds I_{n_i}}P_ix_i \\
	&= x_i^\top P_i\begin{bmatrix}
		\mathds I_{n_i}\\
		S_i^\top
	\end{bmatrix}^{\!\!\top}\!\!\begin{bmatrix}
			P_i^{-1} & \bbzero_{n_i\times s_i}\\
			\bbzero_{s_i\times n_i} & \bbzero_{s_i\times s_i}
	\end{bmatrix}\!\!\begin{bmatrix}
		\mathds I_{n_i}\\
		S_i^\top
	\end{bmatrix} P_i x_i,
\end{align*}
and defining $\Lambda_i\coloneq P_i^{-1}$, enforcing condition~\eqref{eq:main con} yields~\eqref{eq:S-procedure}, and therefore~\eqref{eq:new H} holds. Consequently,
	\begin{align*}
		\mathcal L V_i(x_i) &\leq -\kappa_iV_i(x_i)+\rho_i |w_i|_2^2,
	\end{align*}
which completes the proof.\hfill$\blacksquare$

\end{document}